\documentclass[PRA,twocolumn,
]{revtex4-1}
\usepackage{graphicx}
\usepackage{amsfonts}
\usepackage{amssymb}
\usepackage{amsthm}
\usepackage{array}
\usepackage{amsmath}
\usepackage{verbatim} 
\usepackage{hyperref}
\usepackage{color}
\usepackage{bbold}
\usepackage{epstopdf}
\usepackage{mathtools}
\usepackage{enumerate}
\usepackage{subcaption}
\usepackage{color}
\usepackage[normalem]{ulem}

\newtheorem{proposition}{Proposition}
\newtheorem*{proposition*}{Proposition}
\newtheorem{definition}{Definition}
\newtheorem*{definition*}{Definition}
\newtheorem{theorem}{Theorem}
\newtheorem*{theorem*}{Theorem}

\newtheorem*{corollary*}{Corollary}

\newtheorem{lemma}{Lemma}

\newcommand{\ket}[1]{\left\vert#1\right\rangle}
\newcommand{\bra}[1]{\left\langle#1\right\vert}

\newcommand{\abs}[1]{\left|#1\right|}

\def\Tr{\mbox{Tr}}

\begin{document}
\title{Coherence, Quantum Fisher Information, Superradiance and Entanglement are Interconvertible Resources}
\author{Kok Chuan Tan}
\email{bbtankc@gmail.com}
\author{Seongjeon Choi}
\author{Hyukjoon Kwon}
\author{Hyunseok Jeong}
\email{jeongh@snu.ac.kr}
\affiliation{Center for Macroscopic Quantum Control \& Institute of Applied Physics, Department of Physics and Astronomy, Seoul National University, Seoul, 151-742, Korea}
\date{\today}

\begin{abstract}
We demonstrate that quantum Fisher information and superradiance can be formulated as coherence measures in accordance with the resource theory of coherence, thus establishing a direct link between metrological resources, superradiance and coherence. The arguments are generalized to show that coherence may be considered as the underlying fundamental resource for any functional of state that is first of all faithful, and second, concave or linear. It is also shown that quantum Fisher information and the superradiant quantity are in fact antithetical resources in the sense that if coherence were directed to saturate one quantity, then it must come at the expense of the other. Finally, a key result of the paper is to demonstrate that coherence, quantum Fisher information, superradiant quantity, and entanglement are mutually interconvertible resources under incoherent operations. 
\end{abstract}

\maketitle

\section{Introduction}

The study of quantum resources has seen another revival of interest over the last several years due to the recent identification and characterization of a resource theory of coherence~\cite{Baumgratz2014}. While the coherence of quantum systems has always, in some form or another, been recognized as a fundamental aspect of the field, the newly developed resource theory now provides a framework that allows for a more quantitative understanding of the subject. Since this development, coherence has now been studied within the contexts of quantum correlations~\cite{Tan2016, Streltsov15,  Ma2016}, macroscopic quantumness~\cite{Yadin2015, Kwon2017}, nonclassical light~\cite{Tan2017,Zhang2016,Xu2016}, interferometric experiments~\cite{Wang2017}, error correction~\cite{Tan2017-2}, quantum estimation~\cite{Giorda2016}, and quantum algorithms~\cite{Hillery2016, Matera2016}. There are also several different variations of the theory~\cite{Chitambar2017}, such as a recent proposal for a resource theory of superposition~\cite{Theurer2017} which generalizes the concept of coherence. An extensive overview of the subject can be found at~\cite{Streltsov2016}.

An area that has also garnered considerable interest concerns the convertibility of coherence into nonclassical correlations such as entanglement~\cite{Streltsov15,Ma2016, Killoran2016,Regula2017}. Already, an experimental conversion of coherence to quantum correlations and vice versa has been recently reported~\cite{Wu2017}. Given that such quantum correlations often find applications in a variety of scenarios, the study of such interconversion processes allow for greater flexibility when extracting practical advantages out of nonclassical quantum resources. 

In this paper, we further explore these ideas by demonstrating how quantum Fisher information (QFI), superradiance, entanglement and coherence may be related to each other via coherence. A key result of the paper is to demonstrate that coherence, QFI, superradiant quantity, and entanglement are in fact mutually interconvertible resources under incoherent operations. A central theme of these results is the optimal application of non coherence increasing operations, otherwise called incoherent operations, on quantum states.

\section{Preliminaries}

We first review some elementary concepts concerning coherence measures. The notion of coherence that we will employ in this paper will be the one identified in~\cite{Baumgratz2014}, where a set of axioms is identified in order to specify a reasonable measure of quantum coherence. The axioms are as follows:

For a given fixed basis $\{ \ket{i} \}$, the set of incoherent states $\cal I$ is the set of quantum states with diagonal density matrices with respect to this basis. Incoherent completely positive and trace preserving maps (ICPTP)
 are maps that map every incoherent state to another incoherent state. Given this, we say that  $\mathcal{C}$ is a measure of quantum coherence if it satisfies following properties:
(C1) $\mathcal{C}(\rho) \geq 0$ for any quantum state $\rho$ and equality holds if and only if $\rho \in \cal I$.
(C2a) The measure is non-increasing under a ICPTP map $\Phi$ , i.e., $\mathcal{C}(\rho) \geq \mathcal{C}(\Phi(\rho))$.
(C2b) Monotonicity for average coherence under selective outcomes of ICPTP:
$\mathcal{C}(\rho) \geq \sum_n p_n \mathcal{C}(\rho_n)$, where $\rho_n = K_n \rho K_n^\dagger/p_n$ and $p_n = \Tr [K_n \rho K^\dagger_n ]$ for all $K_n$ with $\sum_n K^\dagger_n K_n = \mathbb 1$ and $K_n {\cal I} K_n^\dagger \subseteq \cal I$.
(C3) Convexity, i.e. $\lambda \mathcal{C}(\rho) + (1-\lambda) \mathcal{C}(\sigma) \geq \mathcal{C}(\lambda \rho + (1-\lambda) \sigma)$, for any density matrix $\rho$ and $\sigma$ with $0\leq \lambda \leq 1$.

One may check that a particular operation is incoherent if each of its Kraus operators\cite{kraus} always maps a diagonal density matrix to another diagonal density matrix. One important example of such an operation is the CNOT gate.

\section{Coherence and Quantum Fisher Information}

 We now consider a standard metrological scenario. One first begins with the signal Hamiltonian, which is denoted $\theta H_S$ . The signal Hamiltonian encodes a signal on a probe state, which is a specially prepared quantum state $\rho$ of $N$ particles, or more if one were to include any possible ancillary quantum particles. The Hamiltonian generates the dynamics $\rho(t) = e^{-iH_S \theta t} \rho e^{iH_S \theta t}$ and after some time $t = \tau$, a measurement is then performed on the state $\rho(\tau)$, the outcome of which is specifically designed in order to obtain the most precise estimate of the value of $\theta$. The optimal sensitivity is known to be given by the quantum Cram{\'e}r-Rao bound~\cite{Braunstein1994} $\delta\theta \geq \frac{1}{\sqrt{\nu\mathcal{F}(\rho(\tau),H_S)}}$ where $\nu$ is the number of measurements performed and $\mathcal{F}(\rho(\tau),H_S)$ is the QFI of a state with respect to $H_S$, given by 
$$
\mathcal{F}(\rho(\tau),H_S) = 2\sum_{i,j}\frac{(\lambda_i - \lambda_j)^2}{\lambda_i + \lambda_j } |\langle i | H_S |j\rangle|^2,
$$
where $\lambda_i$ and $\ket{i}$ are eigenvalues and eigenstates of $\rho(\tau)$, respectively. We are primarily interested in the sensitivity of the state $\rho$ locally at the point $t = 0$, so $\mathcal{F}(\rho,H_S)$ will be the figure of merit we will consider.

A class of Hamiltonians of particular interest is the class of  \textit{local Hamiltonians}. These Hamiltonians are a sum of $N$ independent Hamiltonians acting on individual particles, i.e. a Hamiltonians of the form $H_S = \sum_{i=1}^N h^{(i)}$ where $h^{(i)}$ represents a nontrivial interaction acting only on the $i$th particle that is not proportional to the identity. We will also assume that $h^{(i)}$ does not depend on the number $N$.  An example of a Hamiltonian of this type is a uniform magnetic field in the $z$ direction acting on a collection of $N$ spins, where in this case $h^{(i)} \propto \sigma_z^{(i)}$, and $\sigma_z^{(i)}$ are the Pauli $z$ operators acting on $i$th particle. As coherence is a basis dependent concept, we will adopt the basis which is naturally defined by the eigenvectors of $h^{(i)}$. 
This defines a set of local bases $\{ \ket{a^{(i)}} \}$ for the $i$th particle where $a^{(i)} = 1, \ldots, d$, and $d$ is the dimension of the $i$th particle. Consequently, we will consider the coherence with respect to this set of local bases. Local bases were also previously studied in~\cite{Tan2016}, which noted their connection with quantum correlations. 


For any signal Hamiltonian of the form $H_S = \sum_{i=1}^N h^{(i)}$, and a pure state probe $\ket{\psi}$, let us consider the maximal QFI reachable via all possible incoherent operations $\Phi \in \mathrm{ICPTP}$ on $\ket{\psi}$, i.e. $\displaystyle \max_{\Phi \in \mathrm{ICPTP} }\mathcal{F}(\Phi(\ket{\psi}\bra{\psi}),H_S)$. The incoherent operation here is completely general, with no constraints otherwise. Here, we note that there is an important differentiation between $N$, which captures the number of particles $H_S$ is interacting with, and the actual physical number of particles, which can be any arbitrary number so long as it is reachable via an incoherent operation. 

In fact, for any coherent pure state, we can always achieve Heisenberg scaling via a suitable incoherent operation, as demonstrated by the following Lemma:

\begin{lemma} \label{lem::purestates}
For every coherent pure state $\ket{\psi}$ and locally interacting Hamiltonian $H_S$, there always exists an incoherent operation $\Phi$ that achieves $\mathcal{F}(\Phi(\ket{\psi}\bra{\psi}),H_S) > 0$ that scales with $\mathcal{O}(N^2)$. The measurement that achieves this Heisenberg limited scaling can also be performed incoherently. 
\end{lemma}

\begin{proof}

Let us first consider $H_S = \sum_{j=1}^N h^{(j)}$. For each $h^{(j)}$, let $\ket{i^{(j)}_{max}}$ and $\ket{i^{(j)}_{min}}$ be eigenvectors that corresponds to eigenvectors for the maximum and minimum eigenvalues $\lambda_{\rm max} (h^{(j)})$ and $\lambda_{\rm min} (h^{(j)})$, respectively. In this 2 dimensional subspace, let us define the Pauli operator $\sigma_x^{(j)} \coloneqq \ket{i^{(j)}_{max}}\bra{i^{(j)}_{min}} + \ket{i^{(j)}_{min}}\bra{i^{(j)}_{max}}$. 

Let $\ket{\psi} = \sum_i \sqrt{\lambda}_i \ket{i}$,  where $\ket{i}$ are eigenstates of $H_S$ which construct the incoherent basis. Without any loss in generality, we assume that the coefficients are positive real and $\lambda_1 \geq \lambda_2 \geq \ldots $. We will also assume that $\ket{i=1} = \ket{i^{(1)}_{max}}$ and $\ket{i=2} = \ket{i^{(1)}_{min}}$ since this is just a relabelling of the basis which can be done using an incoherent unitary. The `extra' particles may be considered ancillary particles that assist during the metrological process.

We now apply an incoherent CNOT type operation that performs the map $U : \ket{\psi} \rightarrow \sum_i \sqrt{\lambda}_i \ket{i \ldots i} $  and then let the state evolve according to the Hamiltonian $H_S$. Let us now consider only the the first 2 terms of $U\ket{\psi}$, which under $H_S$ evolves as $$ \sqrt{\lambda}_1 \ket{1 \ldots 1} +  \sqrt{\lambda}_2 \ket{2 \ldots 2} \rightarrow \sqrt{\lambda}_1 \ket{1 \ldots 1} +  \sqrt{\lambda}_2 e^{i\phi\tau}\ket{2 \ldots 2}$$ up to an overall phase factor. We have $\phi = \sum_{j=1}^N \phi_j $ where $\phi_j \coloneqq \lambda_{max}(h^{(j)})-\lambda_{min}(h^{(j)})$.

We will choose some basis on the Hilbert space space of $N$ particles $\{\ket{a_i} \}$ for $i=1,2, \ldots$ such that $\ket{a_1} = \frac{1}{\sqrt2}(\ket{1\ldots1}+\ket{2 \ldots 2})$. Define the following POVM:

$$
M_{(\vec{c},i)} \coloneqq \ket{ \pi(\vec{c}) , i}  \prod_{j=1}^N \bra{(-)^{c_j}} a_i\rangle\bra{a_i},
$$
where $\vec{c} \coloneqq (c_1,\ldots , c_N)$, $c_j = 0,1$, and $\pi(\vec{c}) = \prod^{N}_{j=1}(-1)^{c_j}$. The quantum operation $\mathcal{M}$ is then defined as $\mathcal{M}(\rho) = \sum_{(\vec{c},i)} M_{(\vec{c},i)} \rho M^\dag_{(\vec{c},i)}$. This operation is incoherent and is effectively an incoherent implementation of 2 measurements: a projection onto the basis $\{ \ket{a_i} \}$ followed by a parity measurement on the $x$ axis. Suppose we perform the naive protocol where if the measurement outcome is $i=1$, we keep the parity measurement outcome, and assign a value of zero otherwise. Let us call this measurement $\mathcal{M}'$. 

We can then verify using the error propagating formula that $$\frac{\Delta\mathcal{M}'^2}{\abs{\partial_\tau \langle \mathcal{M}' \rangle }^2} = \frac{(\sqrt{\lambda_1}+\sqrt{\lambda_2})^2}{2}\frac{1}{\phi^2}$$

Finally, we observe that $\phi_i \geq \phi_{\min} \coloneqq \min_j \phi_j$ and $\phi \geq N\phi_{\min}$. $\phi_{\min}$ depends only on $\{h^{(i)}\}$, all of which do not contain any dependence on $N$, and neither does the coefficient $\frac{(\sqrt{\lambda_1}+\sqrt{\lambda_2})^2}{2}$, which depends only on the initial state. As such, we have $$\frac{\Delta\mathcal{M}'^2}{\abs{\partial_\tau \langle \mathcal{M}' \rangle }^2} \leq \frac{(\sqrt{\lambda_1}+\sqrt{\lambda_2})^2}{2}\frac{1}{N^2 \phi_{\min}} \sim \mathcal{O}(\frac{1}{N^2})$$

This proves that for every pure coherent state, Heisenberg limited scaling is reachable using only incoherent operations.
\end{proof}

A natural consequence of the Lemma~1 is that the trivial coherence measure defined by $\mathcal{C}(\rho) = 1$ iff $\rho$ is coherent acquires a metrological interpretation. As this fact is not one of the main themes of this article, a  short discussion of this will be deferred to the Appendix.

So far, we have only considered pure states and a single metrological experiment. However, we can also consider the case of general mixed states where $M$ independent measurements are performed: 

%
%
%
%
%
%

\begin{definition} [Distributed coherence of QFI] \label{thm::monotonicity}
The distributed QFI for a pure state $\ket{\psi}$ is defined to be 
\begin{align*}
&\mathcal{C}^{M}_{F}(\ket{\psi}) \coloneqq \max_{\Phi \in \mathrm{ICPTP}}\sum_{i=1}^M \mathcal{F}\{ \mathrm{Tr}_{ P(i)^c}[\Phi(\ket{\psi}\bra{\psi})], H_S^{(i)}\}
\end{align*} where $H_S^{(i)}$ is the $i$th local Hamiltonian of the form $H_S^{(i)} \coloneqq \sum_{j=1}^N h^{(i,j)}$ and $h^{(i,j)}$ are nontrivial. $P(i)$ refers to the $i$th partition of particles in the state $\Phi(\ket{\psi}\bra{\psi})$ which is partitioned into $M$ collections of particles that separately interact with the Hamiltonians $H^{(i)}_S$. The partial trace $\mathrm{Tr}_{ P(i)^c}$ is to be interpreted as tracing out every particle except the ones in $P(i)$.

The generalization to mixed states is obtained via the convex roof construction 
$$
\mathcal{C}^{M}_F(\rho) \coloneqq \min_{\{p_i, \ket{\psi_i}\}} \sum_{i} p_i \mathcal{C}^{M}_{F}(\ket{\psi_i})
$$ where the minimization is over all pure state decompositions of the form $\rho = \sum_i p_i \ket{\psi_i}\bra{\psi_i}$. (See~\cite{Yuan2015} for another example of such constructions)
\end{definition}   

This definition corresponds to a scenario where a quantum state $\Phi(\ket{\psi}\bra{\psi})$ is prepared via an incoherent operation, partitioned into $M$ separate subsystems, and then distributed to $M$ different parties each performing an independent metrological experiment. Equivalently, it can also be interpreted as a single party scenario where $M$ independent metrological experiments are performed.

It turns out that $\mathcal{C}_F^M$ is a valid coherence measure for every $M$ and $H_S^{(i)}$. 

\begin{theorem} \label{thm::distFish}
$\mathcal{C}^M_F$ is a coherence measure.
\end{theorem}

\begin{proof}
We observe that if $\rho$ is incoherent, then it is diagonal w.r.t. $\sum_i H_S^{(i)}$, and  $\mathcal{F}(\rho,H_S^{(i)}) = 0$. Resorting to any incoherent operation $\Phi$ will not improve the situation as it always maps a diagonal state to another diagonal state so we must have have $\mathcal{C}^M_F(\rho) = 0$. Lemma~\ref{lem::purestates} then demonstrates that  if $\rho$ is coherent, then $\mathcal{C}^M_F(\rho) > 0$ since $\rho$ has to have at least one pure state in its pure state decomposition that is coherent. This proves that $\rho$ is incoherent iff $\mathcal{C}^M_F(\rho) = 0$, so the measure is faithful. 

Convexity is implied by the convex roof construction. Therefore, we only need to prove strong monotonicity. 

To prove monotonicity, we only need to establish that the measure is strongly monotonic over pure states (a short proof of this fact is presented in the Appendix). We see that this is true from the following chain of inequalities:

\begin{align}
& \sum_i p_i\mathcal{C}_F^M(\ket{\psi_i}) \\
&= \sum_ip_i \mathcal{C}_F^M(K_i \ket{\psi} / \sqrt{p_i}) \\
&=  \sum_{i} p_i  \max_{\Phi_{i} \in \mathrm{ICPTP}}\sum_{k=1}^M \notag\\ & \quad \mathcal{F}\{ \mathrm{Tr}_{P(k)^c}[\Phi_{i}(K_i\ket{\psi}\bra{\psi}K^\dag_i/p_i)], H_S^{(k)}\}\\
&=   \max_{\Phi_{i} \in \mathrm{ICPTP} }\sum_{i}p_i
\sum_{k=1}^M \mathcal{F}\{\mathrm{Tr}_{P(k)^c}
\notag\\&\quad[\Phi_{i}(K_i\ket{\psi}\bra{\psi}K^\dag_i/p_i)\otimes \ket{i}\bra{i}],  H_S^{(k)}\otimes  \ket{i}\bra{i}\} \\
&=   \max_{\Phi_{i} \in \mathrm{ICPTP} }
\sum_{k=1}^M \mathcal{F}\{\mathrm{Tr}_{P(k)^c}
\notag\\&\quad[\sum_{i}\Phi_{i}(K_i\ket{\psi}\bra{\psi}K^\dag_i)\otimes \ket{i}\bra{i}], \notag\\ & \quad \sum_{i}H_S^{(k)}\otimes \ket{i}\bra{i}\} \\
&\leq \max_{\Phi \in \mathrm{ICPTP} }
\sum_{k=1}^M \mathcal{F}\{\mathrm{Tr}_{P(k)^c}[\Phi(\ket{\psi}\bra{\psi})], H_S^{(k)}\} \\
&=\mathcal{C}^M_F(\rho)
\end{align} where the inequality in Line 6 comes from the observation that the optimization in Line~5 is a special case of the optimization over $\Phi$ in Line~7.

\end{proof}

Note that it is possible to generalize the result to arbitrary signal Hamiltonians rather than the local Hamiltonians which is the focus of this article. To see this, simply set $N = 1$ so $H_S = h^{(1)}$ where $h^{(1)}$ is in principle any arbitrary nontrivial Hamiltonian. Nonetheless, the case of local Hamiltonians is interesting due to its connections with multipartite quantum correlations. 

We also note that Thm.~\ref{thm::monotonicity} applies for every $H_S$ and $M$. In particular, for the case $M=1$, $\mathcal{C}^1_F(\rho)$ is just the standard Fisher information, optimized over all incoherent operations performed on the state $\rho$. However, a closer inspection will reveal that for small $M$, the measure will saturate for relatively slow levels of coherence. We already see this from the fact that for $M=1$ a maximally coherent qubit can already be converted to a GHZ state via a series of CNOT operations, which is sufficient to saturate the QFI and $\mathcal{C}^1_F$ for $H_S = \sum_i \sigma_z^{(i)}$. As such, depending on the system being considered, larger values of $M$ may lead to better coherence measures.

\section{Coherence and supperradiance}

In this section, we demonstrate that the effect chiefly responsible for the supperradiant phenomena can also be attributed to coherence. 

We make some necessary definitions. Consider a system consisting of $N$  subsystems with an excited and a ground state denoted $\ket{e}$ and $\ket{g}$ respectively. We note that this does not necessarily imply that the state is composed of $N$ two level systems, only that the optical transition between these two states out of a possible $d$ levels for each of the $N$ subsystems are addressed. This optical transition corresponds to wavelength $\lambda$. We define the raising and lowering operators acting on the $i$th subsystem as $D^{(i)}_+ \coloneqq \ket{e^{(i)}}\bra{g^{(i)}}$ and $D^{(i)}_- \coloneqq \ket{g^{(i)}}\bra{e^{(i)}}$ respectively. 

In standard florescence, it is assumed that each of these two level systems interacts independently with the radiation field, during which the total rate of photon emission is simply the sum $\sum_i W^{(i)}_1 \propto \sum_i \langle D^{(i)}_+D^{(i)}_- \rangle$. This implies the emission rate at most scales with $N$, which is intuitively the maximum possible number of excited states at any moment.

In the superradiant regime, it is assumed that the linear dimension of the $N$ systems is small with respect to $\lambda$, so there is a collective, coherent interaction with the radiation field. In this case, the $N$ systems collectively behave like a single dipole, in which case the emission rate is \cite{Gross1982} $$W_N \propto \langle \sum_{i} D^{(i)}_+ \sum_{j} D^{(j)}_- \rangle = \sum_i \langle D^{(i)}_+D^{(i)}_- \rangle + \sum_{i\neq j} \langle D^{(i)}_+D^{(j)}_- \rangle.$$

We see that the second to last term $\sum_i \langle D^{(i)}_+D^{(i)}_- \rangle$ is the sum of single system emissions. The last term $\sum_{i\neq j} \langle D^{(i)}_+D^{(j)}_- \rangle$ is due to the collective behaviour of the subsystems and the source of supperradiant phenomena, which can potentially scale with $N^2$. We will refer to this as the \textit{superradiant quantity}.

In a similar vein as the case for QFI, we consider the following quantity: $$\mathcal{C}_{S}(\ket{\psi})\coloneqq \max_{\Phi \in \mathrm{ICPTP}} \sum_{i \neq j}^N\mathrm{Tr} [ \Phi(\ket{\psi}\bra{\psi})D^{(i)}_+D^{(j)}_- ],$$ which is essentially the maximal superradiant quantity that is achievable via an incoherent operation for the pure state $\ket{\psi}$. Here, the incoherent basis is specified by the excited/ground states for each subsystem.
In this case, the quantity $N$ refers to the number of transitions being addressed. We generalize this to arbitrary mixed state by applying the convex roof construction as before:

\begin{definition}[Coherence of superradiance]

The coherence of superradiance is defined as $$\mathcal{C}_{S}(\rho)\coloneqq \min_{\{p_i, \ket{\psi_i}\}}\sum_{i} p_i\mathcal{C}_{S}(\ket{\psi_i}). $$ The minimization is over all pure state decompositions of the state $\rho = \sum_i p_i \ket{\psi_i}\bra{\psi_i}$.
\end{definition}

Interestingly, it turns out that this, too, is a valid coherence measure.

\begin{theorem} \label{thm::supercoh}
The superradiant coherence $\mathcal{C}_{S}(\rho)$ is a valid coherence measure for every $N \geq 2$.
\end{theorem}

\begin{proof}
We first note that convexity is guaranteed by the convex roof construction, so we just need to prove the faithfulness property and the strong monotonicity property for $\mathcal{C}_{S}$ to be a valid coherence measure.

To prove faithfulness, it is sufficient to demonstrate that it is valid for pure states. It is easy to verify that if some state $\ket{\psi}$ is incoherent, then $\mathrm{Tr} [ \ket{\psi}\bra{\psi}D^{(i)}_+D^{(j)}_- ] = 0 $ for every $i\neq j$. Resorting to an incoherent operation will not help, since that will only lead to a mixture of incoherent pure states, and since the superradiant quantity is a linear functional, this implies that $\mathcal{C}_{S}(\rho)=0$ when $\rho$ is incoherent.

In order to prove that $\mathcal{C}_{S}(\rho)>0$ when $\rho$ is coherent, again, we only need to prove that it is true for pure states due to the convex roof construction. Suppose $\ket{\psi}$ is some coherent state. Then we are guaranteed that $\ket{\psi}= a_0 \ket{0} + a_1 \ket{1} + \ldots $. Without any loss in generality, we will assume that $a_0 \geq a_1 > 0$ and that $\ket{0}$ and $\ket{1}$ corresponds to the ground and excited states respectively. We note here that all the coefficients may be made positive via an incoherent unitary operation. We then perform the following incoherent transformation of state: $$U\ket{\psi}\ket{0} = a_0 \ket{0}\ket{1} + a_1 \ket{1}\ket{0} + \ldots$$ where $U$ is an incoherent unitary operation. We can then directly verify that $\mathrm{Tr} [ \ket{\psi}\bra{\psi}D^{(1)}_+D^{(2)}_- ] = a_0a_1 >0$ so there always exists one incoherent operation that achieves non-zero superradiant quantity for any coherent state $\ket{\psi}$. This implies that $\mathcal{C}_{S}(\rho)>0$ when $\rho$ is coherent, which demonstrates that $\mathcal{C}_{S}(\rho)=0$ iff $\rho$ is incoherent, and proves the faithfulness property.

We now prove strong monotonicity. Here, we also only need to prove it for pure states and the convex roof construction implies it is also true in general (See Appendix). Let the incoherent operation $\Phi$ be the operation with corresponding Kraus operators $K_i$ such that $\Phi(\rho) = \sum_i K_i \rho K_i^\dag$. We note that the map $\Omega(\rho) = \sum_i \Omega_i(K_i \rho K_i^\dag)$ is also a valid incoherent operation as long as $\Omega_i$ is also incoherent for every $i$. Let us assume that $\Omega_i$ is the optimal incoherent operation that achieves the maximal superradiant quantity for $\ket{\psi_i} = \frac{1}{\sqrt{\bra{\psi}K^\dag_iK_i\ket{\psi}}}K_i \ket{\psi}$. We then have the following chain of inequalities:

\begin{align*}
\mathcal{C}_{S}(\ket{\psi}) 
&= \max_{\Phi \in \mathrm{ICPTP}} \sum_{i \neq j}^N\mathrm{Tr} [ \Phi(\ket{\psi}\bra{\psi})D^{(i)}_+D^{(j)}_- ] \\
&\geq  \sum_{i \neq j}^N\mathrm{Tr} [ \Omega(\ket{\psi}\bra{\psi})D^{(i)}_+D^{(j)}_- ] \\
&=  \sum_{i \neq j}^N \sum_k \abs{\bra{\psi}K^\dag_kK_k\ket{\psi}} \mathrm{Tr} [ \Omega_k(\ket{\psi_k}\bra{\psi_k})D^{(i)}_+D^{(j)}_- ] \\
&=\sum_k \abs{\bra{\psi}K^\dag_kK_k\ket{\psi}} \mathcal{C}_{S}(\ket{\psi_k}),
\end{align*}
where the last line is simply the strong monotonicity condition expressed for a pure state. This completes the proof.
\end{proof}

Therefore, superradiant phenomena is also closely related to coherence. In fact, the case of superradiance is suggestive of a much larger class of quantities where any quantum advantage may be directly associated with coherence.

\begin{theorem} \label{thm::concavefun}

Let $f$ be some functional that maps a quantum state to the nonnegative portion of the real line. Then if $f(\ket{\psi})$, or more generally $\displaystyle \max_{\Phi \in \mathrm{ICPTP}}f[\Phi(\ket{\psi})]$, is strictly greater than zero iff $\ket{\psi}$ is coherent, and $f$ is concave or linear, then the convex roof construction $$C_{f}(\rho) \coloneqq \min_{\{p_i, \ket{\psi_i}\} }\max_{\Phi_i \in \mathrm{ICPTP}} \sum_i p_i f[\Phi_i(\ket{\psi_i})]$$ is a valid coherence measure. The minimization is over all pure state decompositions of the state $\rho = \sum_i p_i \ket{\psi_i}\bra{\psi_i}$.

\end{theorem}

\begin{proof}
In order to prove that it is a valid measure, we only need to demonstrate that the above quantity satisfies the strong monotonicity condition for pure states. The faithfulness condition is assumed and convexity comes about naturally due the convex roof construction. 

The proof of strong monotonicity is only lightly modified for the proof of Theorem~\ref{thm::supercoh}. Let the incoherent operation $\Phi$ be the operation with corresponding Kraus operators $K_i$ such that $\Phi(\rho) = \sum_i K_i \rho K_i^\dag$. We note that the map $\Omega(\rho) = \sum_i \Omega_i(K_i \rho K_i^\dag)$ is also a valid incoherent operation as long as $\Omega_i$ is also incoherent for every $i$. Let us assume that $\Omega_i$ is the optimal incoherent operation that achieves the maximal value of $f$ for $\ket{\psi_i} = \frac{1}{\sqrt{\bra{\psi}K^\dag_iK_i\ket{\psi}}}K_i \ket{\psi}$. We then have the following chain of inequalities:

\begin{align*}
\mathcal{C}_{f}(\ket{\psi}) 
&= \max_{\Phi \in \mathrm{ICPTP}} f[\Phi(\ket{\psi}\bra{\psi}]  \\
&\geq  f [\Omega(\ket{\psi}\bra{\psi}]  \\
&=  f\left[ \sum_k \abs{\bra{\psi}K^\dag_kK_k\ket{\psi}} \Omega_k(\ket{\psi_k}\bra{\psi_k})\right] \\
&\geq\sum_k \abs{\bra{\psi}K^\dag_kK_k\ket{\psi}} f[\Omega_k(\ket{\psi_k}\bra{\psi_k})] \\
&=\sum_k \abs{\bra{\psi}K^\dag_kK_k\ket{\psi}} \mathcal{C}_{f}(\ket{\psi_k}).
\end{align*}

The first inequality is due to the maximization over all incoherent operations, of which $\Omega$ is simply one possible candidate, and the second inequality is the to the concavity or linearity of $f$. The last line is simply the expression of strong monotonicity for a pure state. The generalization of strong monotonicity to mixed states is then a natural consequence of the convex roof construction (See Appendix).
\end{proof}

\begin{figure}
	\centering
   	\includegraphics[width=1\linewidth]{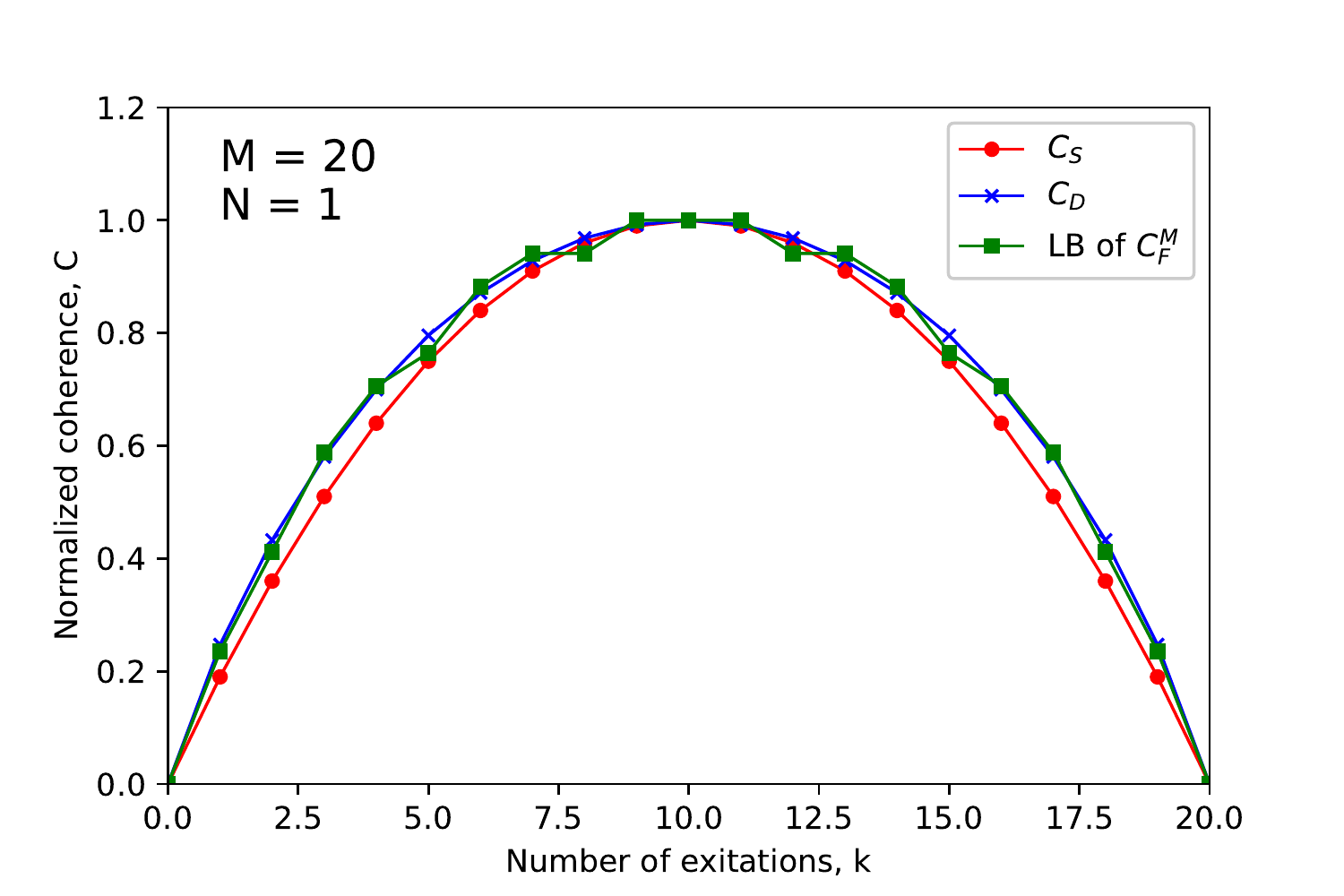}
    \caption{Comparisons between the lower bounds of the superradiant based measure $\mathcal{C}_S$(red, $\circ$), the QFI based measure  $\mathcal{C}_F^M$(green, $\square$), and the relative entropy of coherence $\mathcal{C}_R$(blue, $\times$) for the $M$-qubit Dicke state with $k$ excitations $| M, k\rangle$. The plots are normalized so they coincide at $k=10$, and demonstrate qualititatively similar behaviour across different measures. }
    \label{comparison}
\end{figure}

\section{Examples}

We provide numerical examples comparing our coherence measures $\mathcal{C}_F^M$ and $\mathcal{C}_S$ with the relative entropy of coherence $\mathcal{C}_R$ \cite{Baumgratz2014}. $\mathcal{C}_R$ is defined by $\mathcal{C}_R(\rho) = S(\rho_{\rm diag}) - S(\rho)$ where $S(\rho) = -{\rm Tr}\rho \log \rho$ is the von Neumann entropy of a density matrix $\rho$ and $\rho_{\rm diag}$ is the diagonal part of the density matrix ${\rho}$. While $\mathcal{C}_F^M$ and $\mathcal{C}_S$ requires an optimization over all incoherent operations, and the optimal solution for a general quantum state is in general unknown, any incoherent procedure will provide a lower bound for the measure.

We consider the set of $m$-qubit Dicke state with $k$ excitations. The state is given by
$$ |m, k \rangle = {\binom{m}{k}}^{-\frac{1}{2}} \sum_{P} |P(\underbrace{0 \ldots 0}_\textit{m-k} \underbrace{1 \ldots 1}_\textit{k} )\rangle$$ where $P$ refers to a particular permutation of the state and the summation is over all possible permutations. By verifying the majorization condition \cite{Du2015} for two uniformly distributed pure states, we can show that there exist incoherent operations $\Phi$ which performs the map
$|m, k \rangle \buildrel{\Phi}\over{\longrightarrow} (|0\rangle + |1\rangle)^{\otimes m_k}$ where $m_k$ is the largest integer satisfying $2^{m_k} \leq \binom{m}{k}$. We will use this fact to compute a lower bound for $\mathcal{C}^M_F$ where $M=m$. For the superradiance based measure, we will simply apply the identity operation. Figure~\ref{comparison} illustrates the similarity of the obtained lower bounds to the relative entropy of coherence. Note that the computed values are normalized such that they coincide at the point $k=10$.

\section{QFI and superradiance as opposing resources} 

In the previous sections, we have considered the QFI and superradiance as separate phenomena. In particular, we have considered the optimization of each individual resource w.r.t. some optimal incoherent operations. Here, instead of optimizing for each resource separately, we consider the joint optimization of both quantities. Interestingly, the total sum of distributed QFI and the distributed form of the supperradiant quantity is also itself a coherence measure. 

We first define the coherence of Superradiant-QFI:

\begin{definition} [Coherence of Superradiant-QFI] \label{def::SFcoh}

For pure states, the Coherence of Superradiant-QFI is defined as

\begin{align*}
\mathcal{C}_{SF}(\ket{\psi}) \coloneqq \max_{\Phi \in \mathrm{ICPTP}}   \sum_{i=1}^M  \bigg[ \frac{1}{4} \mathcal{F}\{ \mathrm{Tr}_{P(i)^c}[\Phi(\ket{\psi}\bra{\psi})], H_S^{(i)}\}  \\+ \sum_{k \neq l}^N\mathrm{Tr} [ \Phi(\ket{\psi}\bra{\psi})D^{(i,k)}_+D^{(i,l)}_- ] \bigg],
\end{align*}
where $H_S^{(i)}$ is the $i$th local Hamiltonian of the form $H_S^{(i)} \coloneqq \sum_{k=1}^N h^{(i,k)}$ and $h^{(i,k)}$ are nontrivial. $P(i)$ refers to the $i$th partition of particles in the state $\Phi(\ket{\psi}\bra{\psi})$ which is partitioned into $M$ collections of particles that separately interacts with the Hamiltonians $H^{(i)}_S$. The partial trace $\mathrm{Tr}_{ P(i)^c}$ is to be interpreted as tracing out every particle except the ones in $P(i)$. The operators $D^{(i,k)}_+$ and $D^{(i,k)}_-$ are the raising and lowering operators acting on the $k$th particle in the $i$th partition.  

The extension to mixed states is performed via the convex roof construction $$\mathcal{C}_{SF}(\rho) \coloneqq \min_{\{p_i, \ket{\psi_i}\} } \sum_i p_i \mathcal{C}_{SF}(\ket{\psi_i}).$$ The minimization is over all pure state decompositions of the state $\rho = \sum_i p_i \ket{\psi_i}\bra{\psi_i}$.
\end{definition}

The above is simply the sum of the distributed QFI and the superradiant quantity, except now they are jointly optimized over incoherent operations. As before, the construction is based on considering the optimization for pure states, and the generalization to mixed states is achieved via the convex roof constuction. Assuming that the operators $H_S^{(i)}$ and $D_\pm^{i,k}$ specify the same incoherent basis, we can prove that the joint quantity $\mathcal{C}_{SF}$ is also a valid coherence measure:

\begin{theorem} \label{thm::SFcoh}
$\mathcal{C}_{SF}$ is a valid coherence measure.
\end{theorem}

\begin{proof}
We first note that $\mathcal{C}_{S}$ and $\mathcal{C}^M_F$ are both faithful measures, which implies that there always exists an incoherent operation that ensures $\mathcal{C}_{SF}(\ket{\psi})>0$ iff $\ket{\psi}$ is a coherent pure state. For any mixed coherent state $\rho$, there always exists at least one coherent pure state in its pure state decomposition, which is sufficient to prove that $\mathcal{C}_{SF}(\rho)=0$ iff $\rho$ is an incoherent state. 

Convexity is guaranteed by the convex roof construction, so we only have to prove strong monotonicity for pure states. For a proof of this fact, see Appendix.

To show this, let the incoherent operation $\Phi$ be the operation with corresponding Kraus operators $K_i$ such that $\Phi(\rho) = \sum_i K_i \rho K_i^\dag$. We note that the map $\Omega(\rho) = \sum_i \Omega_i(K_i \rho K_i^\dag) \otimes \ket{i}\bra{i}$ is also a valid incoherent operation as long as $\Omega_i$ is also incoherent for every $i$. Let us assume that $\Omega_i$ is the optimal incoherent operation that achieves the maximal summation of the distributed QFI and the superradiant quantity for $\ket{\psi_i} = \frac{1}{\sqrt{\bra{\psi}K^\dag_iK_i\ket{\psi}}}K_i \ket{\psi}$. We have the following chain of inequalities:

\begin{align}
\mathcal{C}_{SF}&(\ket{\psi})   \\
&\coloneqq\max_{\Phi \in \mathrm{ICPTP}}  \bigg\{  \frac{1}{4}\sum_{i=1}^M \mathcal{F}\{ \mathrm{Tr}_{P(i)^c}[\Phi(\ket{\psi}\bra{\psi})], H_S^{(i)}\}  \notag\\
&+ \sum_{i=1}^M \sum_{k \neq l}^N\mathrm{Tr} [ \Phi(\ket{\psi}\bra{\psi})D^{(i,k)}_+D^{(i,l)}_- ] \bigg\} \\
&\geq \frac{1}{4}\sum_{i=1}^M \mathcal{F}\{ \mathrm{Tr}_{P(i)^c}[\Omega(\ket{\psi}\bra{\psi})], H_S^{(i)}\}  \notag\\
&+ \sum_{i=1}^M \sum_{k \neq l}^N\mathrm{Tr} [ \Omega(\ket{\psi}\bra{\psi})D^{(i,k )}_+D^{(i,l)}_- ] \\
&= \frac{1}{4}\sum_{i=1}^M \mathcal{F}\bigg\{ \mathrm{Tr}_{P(i)^c}[ \notag\\
&\sum_{k'} p_{k'} \Omega_{k'}(\ket{\psi_{k'}}\bra{\psi_{k'}})]\otimes \ket{k'}\bra{k'} ,  H_S^{(i)} \notag\\
&\otimes \sum_{k'} \ket{k'}\bra{k'} \bigg\}  \notag\\
&+ \sum_{i=1}^M \sum_{k \neq l}^N \mathrm{Tr} \left[ \sum_{k'} p_{k'} \Omega_{k'}(\ket{\psi_{k'}}\bra{\psi_{k'}})D^{(i,k)}_+D^{(i,l)}_- \right] \\
&= \sum_{k'} p_{k'} \bigg \{ \frac{1}{4}\sum_{i=1}^M \mathcal{F}\{ \mathrm{Tr}_{P(i)^c}[\Omega_{k'}(\ket{\psi_{k'}}\bra{\psi_{k'}})] ,  H_S^{(i)}\}  \notag\\
&+ \sum_{i=1}^M \sum_{k \neq l}^N \mathrm{Tr} [ ( \Omega_{k'}(\ket{\psi_{k'}}\bra{\psi_{k'}})D^{(i,k)}_+D^{(i,l)}_- ] \bigg \} \\
&= \sum_{k'} p_{k'} \mathcal{C}_{SF}(\ket{\psi_{k'}})
\end{align}

The above largely follows the same themes as in in Theorem~\ref{thm::supercoh}. This demonstrates the strong monotonicity for a pure state, which also implies it is true for mixed states due to the convex roof construction.
\end{proof}

From the above, we see that the sum of the QFI and the superradiant quantity is bounded by the coherence. An immediate implication of this is that in a coherence limited scenario, optimizing one quantity may in turn imply less quantum resources that can be deployed for the other. 

The following theorem explicitly demonstrates this trade{\color{red}-}off between the superradiant quantity and  QFI for the case of of $N$ spin-$1/2$ systems.
\begin{theorem}[Trade-off relation between superradiance and QFI] \label{thm::tradeoff}
For any density matrix $\rho$ in $N$-particle spin-$1/2$ systems the following bound holds:
$$
\sum_{i \neq j} {\rm Tr} [ \rho D^{(i)}_+ D^{(j)}_-] + \frac{1}{4} {\cal F}(\rho, \sum_{i=1}^N \sigma^{(i)}_z) \leq \mu ( N - \mu ) \leq \frac{N^2}{4},
$$
where $\mu$ is the mean excitation number.
\end{theorem}

\begin{proof}
We first note that each pair of $D^{(i)}_+ D^{(j)}_-$ commutes with $\sum_{i=1}^N \sigma^{(i)}_z$, thus it does not changes the excitation number.
This observation leads to the fact that coherence between different excitation number does not contribute to the superradiance.
Then we have ${\rm Tr} [\rho D^{(i)}_+ D^{(j)}_-] =  \sum_{m=0}^N p_m {\rm Tr} [\rho_m D^{(i)}_+ D^{(j)}_-]$, where
$\rho_m = \Pi_m \rho \Pi_m/{p_m}$ and $p_m = {\rm Tr} \rho \Pi_m$ given by the projector $\Pi_m$ onto the excitation number subspace given by $\sum_{i=1}^N \sigma^{(i)}_z = \sum_{m=0}^N \left( m -\frac{N}{2}\right) \Pi_m$.

Furthermore, with  $\sum_{i \neq j} {\rm Tr} [\rho_m D^{(i)}_+ D^{(j)}_-]$ achieves the maximum value $m (N- m)$ when $\rho_m$ is the pure Dicke state of $N$ particles and excitation number $m$. 
Then we get the following bound:
$$
\begin{aligned}
\sum_{i \neq j} {\rm Tr} [ \rho D^{(i)}_+ D^{(j)}_-] &= \sum_{i \neq j} \sum_{m=0}^N p_m {\rm Tr} [\rho_m D^{(i)}_+ D^{(j)}_-] \\
&\leq \sum_{m=0}^N p_m m  (N-m) \\
&= -{\rm Var}\left(\sum_{i=1}^N \sigma^{(i)}_z\right) - \mu^2 + N \mu \\
&\leq - \frac{1}{4} {\cal F}(\rho, \sum_{i=1}^N \sigma^{(i)}_z) + \mu ( N - \mu),
\end{aligned}
$$
where we used the fact that variance is always large than one-quarter of QFI and $\mu= \sum_{i=1}^N {\rm Tr} (m \rho \Pi_m )$ is the mean excitation number. Also note that $\mu ( N - \mu ) = N^2/4$ for $\mu=N/2$, which is the maximum possible value.
\end{proof}

The above inequality suggests that in the single party, $N$ spin scenario, when coherence is distributed such that the superradiant quantity is saturated, then this must come at the expense of QFI. 
Figure~\ref{F_vs_SR} shows that there is no quantum state which can achieve both maximal supperradiant quantity and QFI simultaneously.
This can be shown by considering two extremal cases: a Dicke state with $N/2$ excitation $\ket{N,\frac{N}{2}}$ and the GHZ-state $\ket{{\rm GHZ}} = \frac{1}{\sqrt{2}}(\ket{0}^{\otimes N}+ \ket{1}^{\otimes N})$.
Each state achieves the maximum superradiance $\sum_{i \neq j} {\rm Tr} [ \ket{N,\frac{N}{2}}\bra{N,\frac{N}{2}} D^{(i)}_+ D^{(j)}_- ] = N^2/4$ and the maximum  QFI ${\cal F}(\ket{\rm GHZ}\bra{\rm GHZ}, \sum_{i=1}^N \sigma^{(i)}_z) = N^2$, respectively. In either case, the state saturates the bound, leaving the other   resource to be zero. Every point on the boundary is reachable via an incoherent operations that maps the Dicke $\ket{N,\frac{N}{2}}$ to the state to the $N$ spin GHZ state $\ket{{\rm GHZ}}$ with probability $p$. Note that the total sum of both quantities, corresponding to $\mathcal{C}_{SF}$, is a constant over $0 \leq p \leq 1$. A related clock/work trade-off relation between coherence resources was recently studied within the context of quantum thermodynamics~\cite{KwonTO}. Superradiant phenomena was also previously studied in relation to the skew information~\cite{Santos2016}.

\begin{figure}
	\centering
   	\includegraphics[width=1\linewidth]{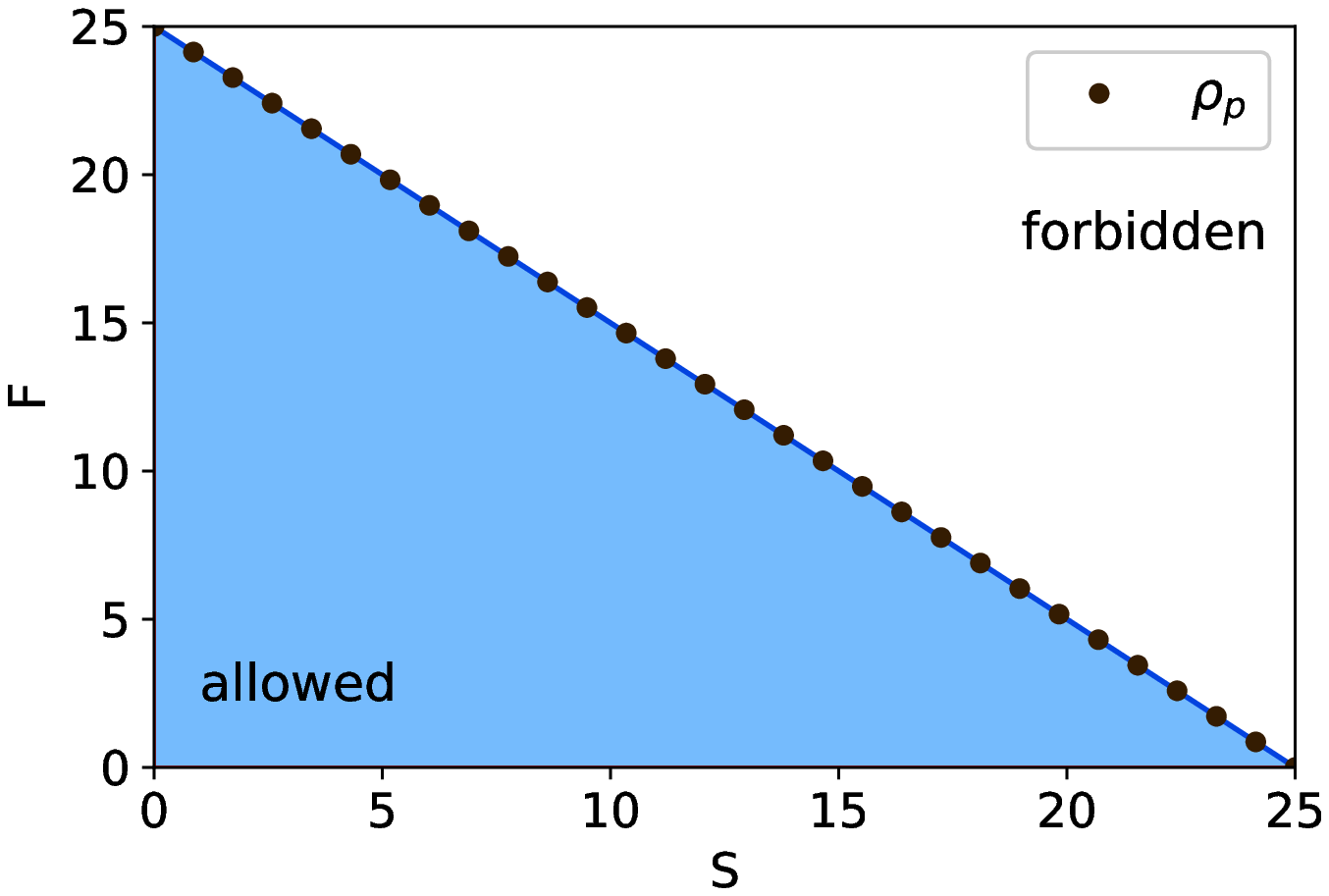}
    \caption{ Trade-off relation between the QFI $F = \frac{1}{4} {\cal F}(\rho, \sum_{i=1}^N \sigma^{(i)}_z)$ and the superradiant quantity $S = \sum_{i \neq j} {\rm Tr} [ \rho D^{(i)}_+ D^{(j)}_-]$ for $N=20$. The shaded region represents the allowed region of $(F, S)$. The boundary (See Inequality~\ref{thm::tradeoff}) can be saturated by a mixture of GHZ state and Dicke state $\rho_{p} = (1-p) \ket{\rm GHZ} \bra{\rm GHZ} + p \ket{N,\frac{N}{2}} \bra{N, \frac{N}{2}}$ (dotted line).
    The maximum achievable $F$ and $S$ are given by the GHZ state and Dicke state, respectively.
}
    \label{F_vs_SR}
\end{figure}

\section{Interconvertibility of resources} In the previous sections, we established strong connections between coherence,  QFI and superradiance. Previous work also established that coherence and entanglement are inconvertible resource via incoherent operations~\cite{Streltsov15}. The following theorem expands upon this by including QFI and superradiance.

\begin{theorem} \label{thm::interconvert}
There always exists an incoherent operation that maps nonzero coherence, QFI, superradiance and entanglement into one another.
\end{theorem}

\begin{proof}
 (QFI/superradiance/entanglement $\rightarrow$ single particle coherence): 
 
First of all, it is immediately apparent that if a state $\rho$ has nonzero QFI, superradiance and entanglement, then $\rho$ must be coherent (See Thms.~\ref{thm::monotonicity} and \ref{thm::supercoh} as well as Ref~\cite{Streltsov15}). Any state $\rho$ of $N$ particles each with dimension $d$, can always be mapped via an incoherent unitary to a single particle state of dimension $d^N$. We can see this already from the 2 qubit case. The basis $\{\ket{00}, \ket{01}, \ket{10}, \ket{11}\}$ can simply be mapped onto the basis $\{\ket{10}, \ket{20}, \ket{30}, \ket{40}\}$ via an incoherent unitary. Its extension to the $N$ particle, $d$ dimensional case is straightforward.

Therefore, it suffices to show that if one has access to a single $d$ dimensional particle that is coherent, then we can achieve nonzero QFI, superradiance and entanglement. 
Furthermore, it also suffices to show the case for a 2 dimensional qubit, since any $d$ dimensional system can be projected to a 2 dimensional qubit via an incoherent operation. This may decrease the coherence, but we simply need to show that there exists at least incoherent operation that converts coherence into QFI, superradiance and entanglement.

(Coherence $\rightarrow$ QFI): Let us consider a general coherent qubit state $\rho = a^2 \ket{1}\bra{1} + ab \ket{1}\bra{0} +ba \ket{0}\bra{1} +b^2 \ket{0}\bra{0}$. Without loss in generality, we can assume that $a,b>0$. If they are complex, we can always perform a unitary that is incoherent to remove the phase.
In this case, it is already immediately clear that $F(\rho, \sigma_z) > 0$, so the identity operation, a trivial incoherent operation, suffices to obtain nonzero  QFI.

(Coherence $\rightarrow$ superrradiance): 
For superradiance, we see that by performing a CNOT with an ancilla initialized in the incoherent state $\ket{1}$, we get \begin{align*} & U_{\mathrm{CNOT}} (\rho\otimes \ket{1}\bra{1}) U_{\mathrm{CNOT}}^\dag \\ & = a^2 \ket{1,0}\bra{1,0} + ab \ket{1,0}\bra{0,1} +ba \ket{01}\bra{10} +b^2 \ket{01}\bra{01}. \end{align*} 

The resulting superradiant quantity is then \begin{align*}\mathrm{Tr}\big[U_{\mathrm{CNOT}} (\rho\otimes \ket{1}\bra{1}) U_{\mathrm{CNOT}}^\dag (D_+^{(1)}D_-^{(2)} + & D_+^{(2)} D_-^{(1)} )\big ]  \\ & = 2ab>0, \end{align*} which is sufficient to show that there exists at least one incoherent operation that converts coherence to superradiance.

(Coherence $\rightarrow$ entanglement):  For entanglement, it is already previously considered in~\cite{Streltsov15}.

Therefore, all 4 quantum resources, coherence, QFI, superradiance and entanglement can be converted into one another using only incoherent operations.
\end{proof}

The above theorem therefore provides a bridge between QFI, superradiance and entanglement via coherence. Figure~\ref{fig::interconvert} illustrates the relationship between the various quantum resources.

It is known that entanglement is a necessary resource in order to achieve Heisenberg limited scaling in QFI~\cite{Pezze2009}. One may therefore be tempted to interpret the relationship between QFI and coherence in terms of the interconvertibility between coherence and entanglement. However, there are entangled states which cannot beat the shot noise limit, even when optimized over local operations~\cite{Hyllus2010}. In contrast, any pure state with nonzero coherence may achieve Heisenberg scaling via some appropriate incoherent operation via Lemma~\ref{lem::purestates}. Furthermore, our results show that any observation of nonzero QFI, even for single systems, is itself a witnessing of a quantum coherence effect, which cannot readily be explained by considering only entanglement.

Similar arguments also show that the Thm.~\ref{thm::supercoh} cannot be readily interpreted only via the conversion of coherence into entanglement. In fact, for the ideal Dicke model of superradiance, although there are many body quantum effects present, no entanglement is actually generated throughout the process~\cite{Wolfe2014}. Finally, the state which maximizes the superradiant quantity is the Dicke state, which yields zero metrological information for a local Hamiltonian of the form $H_S = \sum_i \sigma_z^{(i)}$. The evidence therefore suggests that quantum resources, as specified by QFI, superradiance and entanglement are qualitatively different from one another. They are all, however, connected via incoherent operations, and all possess an interpretation as a form of coherence, suggesting that coherence underlies their operational utility. This relationship is further strengthened in Thm.~\ref{thm::interconvert}, which showed that coherence, QFI, superradiance, and entanglement can be converted into one another via incoherent operations. The exact optimal incoherent operation that performs the conversion processes we leave as an open problem, but it may come in the form of some combination of CNOT type operations, such as the type we see in the faithful conversion of coherence to entanglement~\cite{Regula2017}.

\begin{figure}
	\centering
   	\includegraphics[width=1\linewidth]{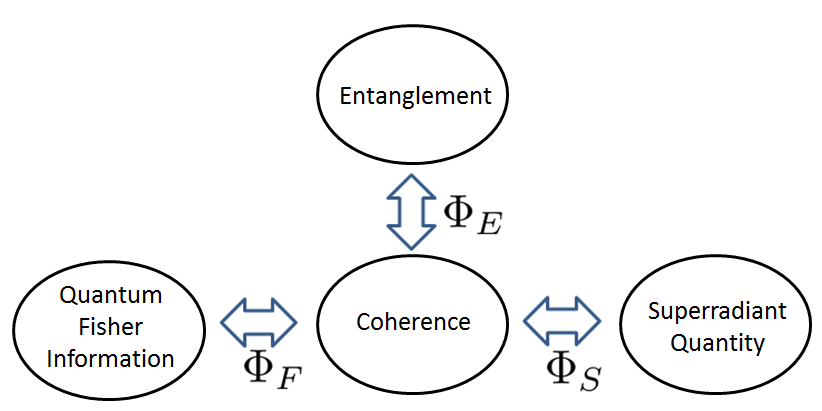}
    \caption{Quantum Fisher Information, superradiant quantity, and entanglement are qualitatively difference resources. However, they are all connected to coherence via incoherent operations. This implies one can move from one resource to another, but only if one considers things from the point of view of coherence. In the diagram $\Phi_F$, $\Phi_S$ and $\Phi_E$ represent incoherent operations which maximize the corresponding resources.}
    \label{fig::interconvert}
\end{figure}

\section{Conclusion}

In the preceding sections, we demonstrated how to construct, via QFI and the superradiant quantity, valid measures of quantum coherence. These measures quantify exactly how much Fisher information, superradiance, or their joint sum is extractable from the coherence in a pure quantum state. The generalization to mixed states may then be achieved via the convex roof construction. In this case, the measures are non-trivial upper bounds that establishes fundamental upper limits to the amount of utility that is extractable without injecting additional coherence into the quantum state. 

It was in fact already known that some form of coherence plays a crucial role in quantum metrology. In \cite{Marvian2016}, it was pointed out that \textit{unspeakable} coherence is especially relevant for metrology, and that the resource theory of asymmetry is able to quantify the metrological usefulness of a given probe state. Our results, in the form of Thms.~\ref{thm::distFish} and \ref{thm::interconvert} goes one step further, by demonstrating that more general forms of coherence may in fact be made useful via an appropriate incoherent operation. In contrast, an operation such as the type considered in Lemma 1 is explicitly forbidden in theories of asymmetry.

We then demonstrate general arguments that are also valid for a large class of functionals of quantum state that is faithful and concave~\ref{thm::concavefun}, demonstrating that a similar interpretation in terms of coherence also exists for such functionals. We also showed that a joint optimization of QFI and the superradiant quantity leads to a valid quantum coherence measure, thus proving an inherent trade{\color{red}-}off between the two processes that is limited by the coherence in the initial state. Finally, we showed that if one were to begin with a state with coherence, QFI, superradiance or entanglement, then there always exists an incoherent operation that converts one resource to another.  

We hope that our work will inspire further research into the role that coherence plays in QFI, Superradiance, quantum correlations, and yet other quantum phenomena.

\section*{Acknowledgments}
This work was supported by the National Research Foundation of Korea (NRF) through a grant funded by the Korea government (MSIP) (Grant No. 2010-0018295) and by the Korea Institute of Science and Technology Institutional Program (Project No. 2E27800-18-P043). K.C. Tan was supported by Korea Research Fellowship Program through the National Research Foundation of Korea (NRF) funded by the Ministry of Science and ICT (Grant No. 2016H1D3A1938100). S. C. was supported by the Global PhD Fellowship Program through the NRF funded by the Ministry of Education (NRF-2016H1A2A1908381).

\appendix
\section{A metrological interpretation of the trivial coherence measure}
In the main text, it was mentioned in passing that the simplest possible coherence measure has a metrological interpretation in terms of Lemma 1. Here, we devote a short discussion illustrating why this is the case.

\begin{definition*} [Asymptotic coherence of QFI]
The Asymptotic coherence of QFI is defined to be the quantity: $$\mathcal{C}_{AF}(\rho) \coloneqq  \max \left (\lim_{N\rightarrow \infty} \frac{\log\mathcal{C}^1_F(\rho)}{\log N}-1, 0\right ). $$ where $N$ is the number of terms in the signal Hamiltonian $H_S = \sum_{i=1}^N h^{(i)}$ and $h^{(i)}$ is nontrivial.
\end{definition*}

From this definition, we can show the following:

\begin{theorem*} \label{thm::metcoh}
$\mathcal{C}_{AF}$ is the trivial coherence measure where $\mathcal{C}_{AF}(\rho)= 1$ iff $\rho$ is coherent, and $\mathcal{C}_{AF}(\rho)= 0$ otherwise.
\end{theorem*}

\begin{proof}
First, we show that that $\rho$ is incoherent iff $\mathcal{C}_{AF}(\rho) = 0$. This is immediately true since no incoherent state can reach the Heisenberg limit even with the help of incoherent operations, so if $\rho$ is incoherent, $\mathcal{C}_{AF}(\rho)=0$. On the other hand, from Lemma~\ref{lem::purestates}, we see that every pure state can reach the Heisenberg limit which implies that if $\rho$ is coherent, $\mathcal{C}_N(\rho) \sim \mathcal{O}(N^2)$ since the pure state decomposition of $\rho$ must contain at least one coherent pure state. This implies that in the limit $N \rightarrow \infty$, $\frac{\log\mathcal{C}_F(\rho)}{\log N} \rightarrow 2$, so  $\mathcal{C}_{AF}(\rho) = 1 >0$. This proves that $\rho$ is incoherent iff $\mathcal{C}_{AF}(\rho) = 0$. Here, we recall that $N$ corresponds to the number of particles that interact with the signal Hamiltonian $H_S$, and not the actual number of physical particles, which can be any arbitrary number so long as it is reachable via some incoherent operation.

We also immediately see that this is just the trivial coherence measure that assigns a value of 1 if a state $\rho$ is coherent, and assigns the value 0 otherwise. It is then easy to verify that the trivial coherence measure satisfies convexity and strong monotonicity, and so is indeed a valid coherence measure in the strict sense. This completes the proof.
\end{proof}

The above therefore provides one physical interpretation for the trivial coherence measure via $\mathcal{C}_{AF}$. We see that by considering the limiting case of $\mathcal{C}^1_F$, the trivial measure corresponds to the fact that every coherent pure state can always achieve Heisenberg limited scaling by applying an appropriate incoherent operation, while for incoherent states the Heisenberg limit is always inaccessible even with the help of arbitrary incoherent operations.

\section{Generalization of strong monotonicity from pure states to mixed states}

A frequently used result in the main paper is the assumption that one only needs to prove the strong monotonicity property for pure states, and the convex roof construction will ensure of its generalization to mixed states. Below, we provide an independent proof of this fact.

\begin{proposition}
Suppose $\mathcal{C}$ is a valid coherence measure over pure states. If we consider the convex roof construction of $\mathcal{C}$

$${\cal C}_{\rm conv.}(\rho) = \min_{ \{ p_i, \ket{\psi_i} \}} \sum_i p_i {\cal C}(\ket{\psi_i}),$$
then ${\cal C}_{\rm conv.}$ is convex, and the strong monotonicity for pure states $ \sum_j q_j {\cal C}(K_j \ket{\psi} / \sqrt{q_j}) \leq {\cal C}(\ket{\psi})$ with $q_j = \bra{\psi} K_j^\dagger K_j \ket{\psi}$ implies the strong monotonicity of ${\cal C}_{\rm conv.}$, i.e. $\sum_j q_j {\cal C}_{\rm conv.}(K_j \rho K_j^\dagger / q_j) \leq {\cal C}_{\rm conv.}(\rho)$ with $q_j = {\rm Tr} (\rho K_j^\dagger K_j)$.
\end{proposition}

\begin{proof}
Convexity: Suppose $\rho = \sum_i p_i \rho_i$ and the optimal decomposition of $\rho_i$ is $\{ q^i_j, \ket{\psi^i_j} \}$ that ${\cal C}_{\rm conv.}(\rho_i) = \sum_j q^i_j {\cal C}(\ket{\psi^i_j})$. Note that $\rho = \sum_{i,j} p_i  q^i_j \ket{\psi^i_j}\bra{\psi^i_j}$ is one of the possible decompositions of $\rho$, thus ${\cal C}_{\rm conv.}(\rho) = {\cal C}_{\rm conv.}(\sum_i p_i \rho_i) \leq \sum_{i,j}  p_i  q^i_j {\cal C}(\ket{\psi^i_j}) = \sum_i p_i {\cal C}_{\rm conv.}(\rho_i) $.

Strong monotonicity: Suppose the optimal decomposition of $\rho$ is $\{ p_\mu^* , \ket{\psi_\mu^*} \}$.
Then $\rho_j = \sum_\mu (p_\mu^* / q^\mu_j) K_j \ket{\psi_\mu^*} \bra{\psi_\mu^*} K_j^\dagger$, where $q^\mu_j = \bra{\psi_\mu^*} K_j^\dagger K_j \ket{\psi_\mu^*}$.
Note that ${\cal C}(\ket{\psi_\mu^*}) \geq \sum_j q^\mu_j {\cal C}(K_j \ket{\psi_\mu^*} / \sqrt{q^\mu_j})$ for all $\mu$. Thus,
$$
\begin{aligned}
{\cal C}_{\rm conv.}(\rho) &= \sum_\mu p_\mu^* {\cal C}(\ket{\psi_\mu^*}) \\
&\geq \sum_{\mu,j} p_\mu^* q^\mu_j {\cal C}\left( \frac{K_j \ket{\psi_\mu^*}}{\sqrt{q^\mu_j}} \right) \\
&\geq  \sum_j q_j {\cal C}_{\rm conv.} (\rho_j),
\end{aligned}
$$
where the last inequality comes from that $\rho_j = \sum_\mu (p_\mu^* q^\mu_j / q_j) \left( \frac{K_j \ket{\psi_\mu^*} \bra{\psi_\mu^*} K_j^\dagger}{q^\mu_j} \right)$ is just one of the possible pure state decompositions of $\rho_j$ and $q_j = {\rm Tr}( \rho K_j^\dagger K_j) = \sum_\mu p^*_\mu q^\mu_j$.
\end{proof}

This proposition demonstrates that we really only need to prove the strong monotonicity case for pure states for convex roof constructions. The demonstration of strong convexity is typically the hardest part of the process, which is greatly simplified by considering only pure states.


\begin{thebibliography}{99}

\bibitem{Baumgratz2014} T. Baumgratz, M. Cramer and M. B. Plenio, Phys. Rev. Lett. {\bf 113}, 140401 (2014).

\bibitem{Tan2016} K.C. Tan, H. Kwon, C.-Y. Park and H. Jeong, Phys. Rev. A {\bf 94}, 022329 (2016).

\bibitem{Streltsov15} A. Streltsov, U. Singh, H. S. Dhar, M. N. Bera and G. Adesso, Phys. Rev. Lett. {\bf 115}, 020403 (2015).

\bibitem{Ma2016} J. Ma, B. Yadin, D. Girolami, V. Vedral and M. Gu, Phys. Rev. Lett. {\bf 116}, 160407.

\bibitem{Yadin2015} B. Yadin and V. Vedral, Phys. Rev. A 92, 022356 (2015).

\bibitem{Kwon2017} H. Kwon, C.-Y. Park, K.C. Tan and H. Jeong, New J. Phys. {\bf 19}, 043024 (2017).

\bibitem{Tan2017} K.C. Tan, T. Volkoff, H. Kwon and H. Jeong, Phys. Rev. Lett. {\bf 119}, 190405 (2017)  

\bibitem{Zhang2016} Y.-R. Zhang, L.-H. Shao, Y. Li and H. Fan, Phys. Rev. A {\bf 93}, 012334 (2016).

\bibitem{Xu2016} J. Xu, Phys. Rev. A {\bf 93}, 032111 (2016).

\bibitem{Wang2017} Y.-T. Wang, J.-S. Tang, Z.-Y. Wei, S. Yu, Z-J. Ke, X.-Y. Xu, C-F. Li and G.-C. Guo, Phys. Rev. Lett. {\bf 118}, 020403 (2017).

\bibitem{Tan2017-2} K.C. Tan, S. Omkar and H. Jeong, arXiv:1704.07572 (2017).

\bibitem{Giorda2016} P. Giorda and M. Allegra, arXiv:1611.02519 (2016). 

\bibitem{Hillery2016} M. Hillery Phys. Rev. A {\bf 93}, 012111 (2016).

\bibitem{Matera2016} J. M. Matera, D. Egloff, N. Killoran and M.B. Plenio Quant. Sci. Tech. {\bf 1}(1), 01LT01 (2016).

\bibitem{Chitambar2017} E. Chitambar and G. Gour, Phys. Rev. A {\bf 95}, 019902(E)
(2017).

\bibitem{Theurer2017} T. Theurer, N. Killoran, D. Egloff and M.B. Plenio, arXiv:1703.10943 (2017).

\bibitem{Streltsov2016} A. Streltsov, G. Adesso and M.B. Plenio, arXiv:1609.02439 (2016).

\bibitem{Killoran2016} N. Killoran, F.E.S. Steinhoff and M.B. Plenio, Phys.
Rev. Lett. {\bf 116}, 080402 (2016).

\bibitem{Regula2017} B. Regula, M. Piani, M. Cianciaruso, T.R. Bromley, A.
Streltsov and G. Adesso, arXiv:1704.04153 (2017).


\bibitem{Wu2017} K.-D. Wu, Z. Hou, Y.-Y. Zhao, G.-Y. Xiang, C.-F. Li, G.-C. Guo, J. Ma, Q.-Y. He, J. Thompson and M. Gu, arXiv:1710.01738 (2017). 

\bibitem{kraus}  K. Kraus : States, Effects and Operations. Springer, Berlin (1983).

\bibitem{Braunstein1994} S. L. Braunstein and C. M. Caves, Phys. Rev. Lett. {\bf 72}, 3439 (1994).

\bibitem{Yuan2015} X. Yuan, H. Zhou, Z. Cao and X. Ma, Phys. Rev. A {\bf 92}, 022124 (2015).

\bibitem{Gross1982} M. Gross and S. Haroche, Phys. Rep. \textbf{93}, p. 301-396 (1982).

\bibitem{Du2015} S. Du, Z. Bai, and Y. Guo, Phys. Rev. A {\bf 91}, 052120 (2015).


\bibitem{KwonTO} H. Kwon, H. Jeong, D. Jennings, B. Yadin, M. S. Kim, Phys. Rev. Lett. {\bf 120}, 150602 (2018).

\bibitem{Santos2016} E.M. dos Santos and E. I. Duzzioni, Phys. Rev. A {\bf 94}, 023819 (2016).

\bibitem{Pezze2009} L. Pezz{\'e} and A. Smerzi, Phys. Rev. Lett. {\bf 102}, 100401 (2009).

\bibitem{Hyllus2010} P. Hyllus, O. G{\"u}hne and A. Smerzi, Phys. Rev. A {\bf 82}, 012337 (2010). 


\bibitem{Wolfe2014} E. Wolfe and S.F. Yelin, Phys. Rev. Lett. {\bf 112}, 140402 (2014).

\bibitem{Marvian2016} I. Marvian and R. W. Spekkens, Phys. Rev. A {\bf 94}, 052324 (2016).


\end{thebibliography}
\end{document}